\documentclass[onecolumn,english,12pt]{IEEEtran}
\usepackage[T1]{fontenc}
\usepackage{babel}
\usepackage{amsmath}
\usepackage{amsthm}
\usepackage{amsfonts}
\usepackage{graphicx}

\makeatletter
\theoremstyle{plain}
\newtheorem{thm}{\protect\theoremname}
\theoremstyle{plain}
\newtheorem{lem}[thm]{\protect\lemmaname}
\theoremstyle{plain}

\newtheorem{prop}{\protect\propositionname}

\theoremstyle{definition}

\usepackage[caption=false,font=footnotesize]{subfig}

\makeatother

\providecommand{\corollaryname}{Corollary}
\providecommand{\lemmaname}{Lemma}
\providecommand{\theoremname}{Theorem}
\providecommand{\examplename}{Example}
\providecommand{\propositionname}{Proposition}
\providecommand{\claimname}{Claim}
\providecommand{\conjecturename}{Conjecture}

\begin{document}
\title{Bounds for Learning Lossless Source Coding}
\author{Anders H{\o}st-Madsen, \emph{Fellow, IEEE}\thanks{A. H{\o}st-Madsen is with the Department of Electrical Engineering, University of Hawaii, Manoa,
 Honolulu, HI, 96822, Email: ahm@hawaii.edu. 
 The research was funded in part by the NSF grant CCF-1908957.}}
\maketitle
\begin{abstract}
This paper asks a basic question: how much training
is required to beat a universal source coder?
Traditionally, there have been two types of source coders: fixed,
optimum coders such as Huffman coders; and universal source coders,
such as Lempel-Ziv
The paper
considers a third type of source coders: learned coders. These are
coders that are trained on data of a particular type, and then used
to encode new data of that type. This is a type of coder
that has recently become very popular for (lossy) image
and video coding.

The paper consider two criteria for performance of learned
coders: the average performance over training data, and
a guaranteed performance over all training except for
some error probability $P_e$. In both cases the coders
are evaluated with respect to redundancy.

The paper considers the IID binary case and binary Markov
chains. In both cases it is shown that the amount
of training data required is very moderate: to code sequences
of length $l$ the amount of training data required
to beat a universal source coder is $m=K\frac{l}{\log l}$,
where the constant in front depends the case considered.
\end{abstract}

\section{Introduction}

Traditionally, there have been two types of source coders: fixed,
optimum coders such as Huffman coders; and universal source coders,
such as Lempel-Ziv \cite{ZivLempel77,ZivLempel78,CoverBook}. We will
consider a third type of source coders: learned coders. These are
coders that are trained on data of a particular type, and then used
to encode new data of that type. Examples could be source coders for
English texts, DNA data, or protein data represented as graphs.

In both machine learning and information theory literatures, there has
been some work on learned coding. From a machine learning perspective,
the paper \cite{SchmidhuberHeil96} stated the problem precisely and
developed and evaluated some algorithms. A few follow up papers, e.g.,
\cite{ZhouFu06,Kattan10,mahoney00,mahoney05,Cox16,Tatwawadi17} have
introduced new machine learning algorithms. 
For lossy coding, in particular of images and video, there has been
much more activity recently, initiated by the
 paper \cite{TodericiAl16} from Google, see for example
\cite{LiChen19,cheng2018deep,LiuLiAl20}. Our aim  is to find
theoretical bounds for how well it is possible to learn coding.
In this paper we will limit ourselves to lossless coding.

From an information theory perspective, Hershkovits and Ziv \cite{HershkovitsZiv97}
considered learned coding in terms of learning a database of sequences.
The results in \cite{HershkovitsZiv97} are quite pessimistic.
Basically they state that to code a sequence of length $l$ so as
to approach the entropy rate $\mathcal{H}$, a length $2^{l\mathcal{H}}$
training sequence is needed -- so that one observes most of the typical
sequences. This means that essentially learned coding is infeasible,
as the amount of training needed is exponential in the sequence length!
Yet, machine learning has shown itself to work very well in other
contexts with large, but not extreme training sets.
We will therefore consider the problem from a different perspective. 

Our perspective on learning coding is to compare with universal
source coders with redundancy as measure. The redundancy of
a coder is is the
difference between the entropy of a source and the average length
achieved by the coder. Suppose that the sources are (or assumed to be)
in some probability class $\Lambda$ characterized by a parameter
vector $\boldsymbol{\theta}$. For a universal source coder
with length function $L$, the redundancy to encode
a sequence of length $l$ is defined by \cite{Shamir06}
\[
R_{l}(L,\boldsymbol{\theta})=\frac{1}{l}E_{\boldsymbol{\theta}}[L(X^{l}]-H_{\boldsymbol{\theta}}(X)
\]
Since $\boldsymbol{\theta}$ is unknown, even in terms of probability law, 
usually the minimax redundancy is considered \cite{Shamir06}
\begin{align*}
R_{l}^{+} & =\min_{L}\sup_{\boldsymbol{\theta}}R_{l}(L,\boldsymbol{\theta})
\end{align*}
A good coder is one that achieves this minimum.

This setup can be generalized to learning.
We are given a training
sequence $x^{m}$; based on the training we develop 
coders $C(x^{l};x^{m})$ with length function $L(x^{l};x^{m})$.
The codelength is $\frac{1}{n}E_{\theta}[L(x^{l};x^{m})|x^{m}]$ (the
expectation here is only over $x^{l}$), and the redundancy is
\begin{equation}
R_{l}(L,x^{m},\theta)=\frac{1}{l}E_{\theta}[L(x^{l};x^{m})|x^{m}]-H_{\theta}(X) \label{eq:Rlearned_single}
\end{equation}
The redundancy depends on the training sequence $x^m$. \emph{One way} to
remove this dependency is to average also over $x^m$,
\begin{align}
R_{l}(L,m,\theta) & =\frac{1}{l}E_{\theta}[L(x^{l};x^{m})]-H_{\theta}(X)
\label{eq:Raverage}\\
R_{l}^{+}(m) & =\min_{L}\sup_{\theta}R_{l}(L,m,\theta) \label{eq:Rlearned}
\end{align}
The idea of learning to code is to obtain information about the distribution
of the source from the training $x^m$ and then apply this to code
the test sequence $x^l$. From a theoretical point of view 
of simply minimizing (\ref{eq:Rlearned}) one could
say that of course the coder can continue to learn about the source
from the test sequence. However, machine learning algorithms usually have a distinct learning
phase, and once the algorithm is trained, it is not updated with test samples.
We will therefore also consider this setup here, and we call such a
learned coder a \emph{frozen coder}. For the coders we consider this
is easy to specify. From the training sequence $x^m$ the coder estimates
a distribution $\hat P(\cdot)$ and this is then applied to test
sequences, without updating, through a Shannon/algebraic coder resulting in
$L(x^l)=-\log(\hat P(x^l))$ (ignoring a possible $+1$).

The question we consider now is: how many training samples do we need in order
to beat a universal source coder, i.e., how large should $m$ be so that
\begin{align}
  R_{l}^{+}(m) \leq R_{l}^{+} \label{eq:criterion}
\end{align}
One might of course want better performance than a universal coder.
But at least one would want the learned coder to do as well
as a universal coder, so (\ref{eq:criterion}) gives a baseline
on performance.

Learning to code has similarities with universal prediction \cite{MerhavFeder98}. The paper \cite{Krichevskiy98}
developed bounds for universal prediction for IID (independent
identically distributed) sources,
and \cite{FalahatgarOrlitsky16,HaoOrlitskyAl18} for
Markov sources. In fact, the paper \cite{Krichevskiy98} exactly
considers (\ref{eq:Rlearned}), and proves
\begin{align}
  \frac{1}{2m\ln 2}+o\left(\frac 1 m\right) &\leq R_{l}^{+}(m)
  \leq \frac{\alpha_0}{m\ln 2}+o\left(\frac 1 m\right) \label{eq:RlIID} \\
  \alpha_0 &\approx 0.50922 \label{eq:alpha0}
\end{align}
The result was improved in \cite{BraessSauer04} to show
that
\begin{align}
  R_{l}^{+}(m) &= \frac{1}{2m\ln 2}+o\left(\frac 1 m\right)
  \label{eq:RlIIDexact}
\end{align}
On the other hand, we also have good expressions for 
$R_{l}^{+}$ \cite{Shamir06}, which can be expressed as $R_{l}^{+}=\frac{\log l}{2l}+o\left(\frac 1 l\right)$. Thus, ignoring $o$-terms, (\ref{eq:criterion}) becomes
\begin{align}
  m\geq \frac{l}{\ln 2\log l} \label{eq:IIDmlimit}
\end{align}
The conclusion is that it is very easy to beat a universal coder.
For $l$ moderately large $\frac{l}{\ln 2\log l}<l$, so we need
fewer training samples than the length of the sequences we 
want to encode.

We now return to (\ref{eq:Raverage}). Averaging over the training
$x^m$ might be reasonable for universal prediction. But in
learning one usually learns once and applies many times. The average
codelength over test sequences in (\ref{eq:Rlearned_single}) is
therefore reasonable, but the averaging over the training less so.
As an alternative one could consider the worst case over $x^m$, but
one can always find totally non-informative training sequences (e.g.,
for the IID case the sequence of all zeros). Instead
one could require that the training is good for most training
sequences, or, put another way, that the probability of a bad
training sequence is low. So, we consider the criterion
\begin{align}
  E(m,a)=\sup_{\boldsymbol{\theta}}P\left(R_{l}(L,x^{m},\theta\right)>a)
  \label{eq:Ema}
\end{align}
For some given $a$ and small $P_e$ the goal is then to ensure
\begin{align*}
  E(m,a) \leq P_e
\end{align*}
Again, we can consider the bottom line of beating the universal
coder, in which case $a=\frac{\log l}{2l}$ for the IID case.

In this paper we consider the measure (\ref{eq:Ema}) for
the IID binary case, and we then generalize to binary
Markov chains for both the average measure (\ref{eq:Rlearned}) and
the error measure (\ref{eq:Ema}).

\section{IID Case}
We consider the IID binary case characterized by the
parameter $\theta=p$, where $p=P(X=1)$ and $q=1-p$. Learning an estimator boils
down to finding an estimator $\hat p$. It is then
well known \cite{CoverBook} that the redundancy of a coder defined
by $\hat p$ is $R_l(L,x^m,\theta)=D(p\|\hat p)$ (except for some
small constant), and therefore
\begin{align*}
	E(m,a) &=\sup_p P\left(D\left(p\|\hat{p}\right)\geq a\right)
\end{align*}

Consider first the
maximum likelihood estimator (MLE), $\hat p=\frac k m$, where
$k$ is the number of ones in the training sequence $x^m$. 
If $k=0$, $D(p\|\hat p)=\infty$, and at the same time
$\lim_{p\to 0}P(k=0)=\lim_{p\to 0}(1-p)^m=1$; therefore
\begin{prop}\label{thm:MLE}
	For the MLE, $E(m,a)=1$ for all $a>0$
\end{prop}
In light of Proposition \ref{thm:MLE} we consider 
other estimators $\hat p$. We assume that $\hat p=f(\check p)$,
where
\begin{align*}
  \check p=\frac k m
\end{align*}
is the minimal sufficient statistic; the function $f$ can
depend on $m$.
For convenience, we assume $f$ is invertible.
Let
\begin{align*}
  P(p,a)&=P\left(D(p\|\hat p)>a\right)
\end{align*}
for fixed $p\leq\frac 1 2$. The equation $D(p\|\hat p)=a$ has two
solutions $\hat p_\pm$ so that
\begin{align*}
  P(p,a)&=P\left(\check p< f^{-1}(\hat p_-)\right)+
  P\left(\check p> f^{-1}(\hat p_+)\right)
\end{align*}
the sum of the \underline{\emph{lower}} and \underline{\emph{upper}} tail probabilities.

We consider what can be named the moderate deviations regime. 
We fix $P_e$ independent of $m$
and require $E(m,a)\leq P_e$ and desire to find 
the smallest $a(m,P_e)$ that satisfies this inequality. We
solve the problem asymptotically as
$m\to\infty$; necessarily $a(m,P_e)\to 0$, and we want
to find how it converges to zero. This essentially gives
the redundancy as a function of $m$. We can use
this to determine how many training samples we need
to beat universal source coding: solving $a(m,P_e)<\frac{l}{2\log l}$.

Let $0\leq\lambda\leq 1$ and put
\begin{align*}
 \hat p_-(p,m,\lambda P_e) &= \inf \{\hat p: P\left(\check p< f^{-1}(\hat p)\right)\leq \lambda P_e\}	\nonumber\\
 \hat p_+(p,m,(1-\lambda)P_e) &= \sup \{\hat p: P\left(\check p> f^{-1}(\hat p)\right)\leq (1-\lambda) P_e\}
\end{align*}
Then we can write
\begin{align}
a(m, P_e) &= \min_\lambda \sup_p	\max\{D(p\| \hat p_-(p,m,\lambda P_e)),D(p\| \hat p_+(p,m,(1-\lambda)P_e))\} \label{eq:am}
\end{align}

For achievability, we consider estimators of the well-known
 form \cite{CoverBook,KrichevskyTrofimov81,Krichevskiy98}
\begin{align}
  \hat p &= \frac{k+\beta m}{m+2\beta m}=\frac{\check p+\beta}{1+2\beta}
  \label{eq:addbeta}
\end{align}
where, as we will see, for moderate deviations we can put 
$\beta = \frac\alpha m$, so that\footnote{The exact result
(\ref{eq:RlIIDexact}) was obtained by using a modified
additive estimator, but we will limit the consideration
here to the plain additive estimator.
}
\begin{align}
  \hat p &= \frac{k+\alpha}{m+2\alpha} \label{eq:addestimator}
\end{align}

The main result  is
\begin{thm}\label{thm:iid_binary}
	For estimators that are functions
	of the sufficient statistic and $P_e$ sufficiently small,
\begin{align}
  a(m,P_e) &\geq \frac{Q^{-1}(P_e/2)^2}{2m\ln 2}+o\left(\frac 1 m\right) \label{eq:ambounds}
\end{align}
The estimator (\ref{eq:addestimator}) has on optimum value
of $\alpha$ that satisfies
\begin{align}
 \frac 1 6 Q^{-1}(P_e/2)^2 -1\leq \alpha \leq \frac 1 6 Q^{-1}(P_e/2)^2 +1\label{eq:alpha_range}
\end{align}
which gives an achievable $a(m,P_e)$;
\begin{align}
  a(m,P_e) &= b(P_e)\frac{Q^{-1}(P_e/2)^2}{2m\ln 2}+o\left(\frac 1 m\right) \label{eq:ambounds}
\end{align}
where
\begin{align*}
  \lim_{P_e\to 0} b(P_e)=1
\end{align*}

\end{thm}
\begin{proof}
We will first argue that we can focus on convergent sequences
in the proof technique.
We want to find the limit $\lim_{m\to\infty}ma(m,P_e)$ (implicitly
a $\limsup$). Let $p_{\max}(m)$ be a value of $p$ where 
$a(m,P_e)$ is achieved -- if there are multiple, we choose
one at random. Consider the sequence $mp_{\max}(m)$; it
has at least one accumulation point when $\infty$ is included.
Again we choose one a random, and a subsequence
$m'p_{\max}(m')$ that converges towards the accumulation
point. It is
clear that 
$\lim_{m'\to\infty}m'a(m',P_e)=\lim_{m\to\infty}ma(m,P_e)$.
We can therefore equivalently find the maximum
of $ma(m,P_e)$ along convergent sequences $mp(m)$. We can
divided such sequences into three regimes:
\begin{itemize}
  \item \emph{CLT regime}: $\lim_{m\to\infty} mp(m) =\infty$.
     In this regime the central limit theorem (CLT) can be
     applied.
  \item \emph{Poisson regime}: $0<\lim_{m\to\infty} mp(m) <\infty$. In this regime a Poisson approximation can be used.
  \item \emph{Sub-Poisson regime}: $\lim_{m\to\infty} mp(m) =0$. 
\end{itemize}
We consider the limit of $ma(m,P_e)$ in each of these regimes,
and maximizes over these limits. In the following
we will drop the explicit dependency $p(m)$ and just write $p$.

\emph{CLT Regime:} We can then use the 
central limit theorem, here for the upper tail,
\begin{align}
P\left(\check p> f^{-1}(\hat p_+)\right)
 &= P\left(\check p-p> f^{-1}(\hat p_+)-p\right) \nonumber\\
 &= P\left(\frac{\sqrt{m}}{\sqrt{pq}}(\check p-p)> \frac{\sqrt{m}}{\sqrt{pq}}( f^{-1}(\hat p_+)-p)\right) \nonumber\\
 &\to Q\left(\lim_{m\to\infty}\frac{\sqrt{m}}{\sqrt{pq}}( f^{-1}(\hat p_+)-p)\right) \quad \text{as }m\to\infty \label{eq:Plim}
\end{align}
(with $Q(x)=1-\Phi(x)$, $\Phi$ being the Gaussian CDF) as Berry-Esseen \cite{SerflingBook} gives
\begin{align}
  \left|P\left(\check p< f^{-1}(\hat p_+)\right)
   -\Phi\left(\frac{\sqrt{m}}{\sqrt{pq}} \left(f^{-1}(\hat p_+) -p\right)\right)\right|
  \leq 0.4748\frac{p^2+q^2}{\sqrt{mpq}} \label{eq:BE}
\end{align}
and the right hand side converges to zero by assumption.
We require
\begin{align*}
  P\left(\check p> f^{-1}(\hat p_+)\right)
  = Q\left(\frac{\sqrt{m}}{\sqrt{pq}}( f^{-1}(\hat p_+)-p)\right)
  +\epsilon(m) \leq (1-\lambda)P_e
\end{align*}
where $lim_{m\to\infty}\epsilon(m)=0$ from (\ref{eq:BE}).
We can include the gap in the inequality in $\epsilon(m)$, so
\begin{align*}
  Q\left(\frac{\sqrt{m}}{\sqrt{pq}}( f^{-1}(\hat p_+)-p)\right)
   = (1-\lambda)P_e+\epsilon(m)
\end{align*}
and thus
\begin{align*}
  \hat p_+ &= f\left(\frac {\sqrt{pq}}{\sqrt{m}}Q^{-1}((1-\lambda)P_e+\epsilon(m))+p\right)
\end{align*}
In the following we will omit the $\epsilon(m)$ as it
does not affect the results.

We will first consider a converse in the CLT regime, more 
specifically for $p$ constant rather than
a function of $m$. This
is clearly also a converse for all regimes.
We  use use Pinsker's inequality for relative entropy \cite{DragomirAl01},
\begin{align}
  D(p\|\hat p_+) &\geq \frac 2{\ln 2} (p-\hat p_+)^2
    \label{eq:Pinsker}
\end{align}
 in (\ref{eq:am})
\begin{align}
  a(m,P_e) \geq \frac 2{\ln 2}\min_{f}\min_\lambda\sup_p\max& \left\{ \left(f\left(\frac {\sqrt{pq}}{\sqrt{m}}Q^{-1}((1-\lambda)P_e)+p\right)-p\right)	^2,\right.\nonumber\\
  &\left.\left(f\left(-\frac {\sqrt{pq}}{\sqrt{m}}Q^{-1}(\lambda P_e)+p\right)-p\right)^2\right\} \label{eq:amlower}
\end{align}
Let $f(x)=x+g_m(x)$, where we have made explicit that $f$ can
depend on $m$. 
We can then write this as
\begin{align}
  a(m,P_e) \geq \frac 2{m\ln 2}\min_{f}\min_\lambda\sup_p\max& \left\{ \left(\sqrt{pq}Q^{-1}((1-\lambda)P_e)+\sqrt{m}g_m\left(\frac {\sqrt{pq}}{\sqrt{m}}Q^{-1}((1-\lambda)P_e)+p\right)\right)	^2,\right.\nonumber\\
  &\left.\left(-\sqrt{pq}Q^{-1}(\lambda P_e)+\sqrt{m}g_m\left(-\frac {\sqrt{pq}}{\sqrt{m}}Q^{-1}(\lambda P_e)+p\right)\right)	^2\right\} \label{eq:amlowerproof}
\end{align}
We will argue that $g_m=0$ and $\lambda=\frac 1 2$ is optimum, or more precisely that
$\lim_{m\to\infty}\sqrt{m}g_m=0$. 
Suppose that for some $p$, $\lim_{m\to\infty}\sqrt{m}g_m\left(\frac {\sqrt{pq}}{\sqrt{m}}Q^{-1}((1-\lambda)P_e)+p\right)=b$, so
that
\begin{align*}
  &\lim_{m\to\infty}\left(\sqrt{pq}Q^{-1}((1-\lambda)P_e)+\sqrt{m}g_m\left(\frac {\sqrt{pq}}{\sqrt{m}}Q^{-1}((1-\lambda)P_e)+p\right)\right)^2 \\
  &= \left(\sqrt{pq}Q^{-1}((1-\lambda)P_e)+b\right)^2
\end{align*}
Let $p_m$ be the solution to
\begin{align*}
  -\frac {\sqrt{p_mq_m}}{\sqrt{m}}Q^{-1}(\lambda P_e)+p_m=\frac {\sqrt{pq}}{\sqrt{m}}Q^{-1}((1-\lambda)P_e)+p.
\end{align*}
Then
\begin{align*}
  &\lim_{m\to\infty}\left(-\sqrt{p_mq_m}Q^{-1}(\lambda P_e)+\sqrt{m}g_m\left(\frac {\sqrt{p_mq_m}}{\sqrt{m}}Q^{-1}((1-\lambda)P_e)+p_m\right)\right)^2 \\
  &= \left(-\sqrt{pq}Q^{-1}(\lambda P_e)+b\right)^2
\end{align*}
Thus the maximum in (\ref{eq:amlowerproof}) as $m\to\infty$ becomes
\begin{align}
  \max\left\{\left(\sqrt{pq}Q^{-1}((1-\lambda)P_e)+b\right)^2,\left(-\sqrt{pq}Q^{-1}(\lambda P_e)+b\right)^2\right\} \label{eq:limmax}
\end{align}
Since there is a minimization over $f$ and $\lambda$, we
can choose $\lambda$ and $b$. It is now easily seen that
(\ref{eq:limmax}) is minimized for $\lambda=\frac 1 2$ and
$b=0$ as follows. We can assume that $P_e<\frac 1 2$ so that
$Q^{-1}(P_e)>0$. For $b=0, \lambda=\frac 1 2$ the two parts of
the max are equal. If we make $b>0$ we must decrease $\lambda$ below $\frac 1 2$
to get $\left(\sqrt{p q}Q^{-1}((1-\lambda)P_e)+b\right)^2
<\left(\sqrt{p q}Q^{-1}(P_e/2)\right)^2$. But for such
$\lambda$, $\left(-\sqrt{p q}Q^{-1}(\lambda P_e)+b\right)^2
>\left(\sqrt{p q}Q^{-1}(P_e/2)\right)^2$ due to the convexity
of $Q^{-1}(x)$ for $x<\frac 1 2$.


Thus, in (\ref{eq:amlower}) the minimum is achieved for
$f$ the identity and $\lambda=\frac 1 2$, while the
maximum over $p$ is achieved for $p=q=\frac 1 2$. This
gives (\ref{eq:ambounds}) as a lower bound.

For achievability in the CLT regime we explicitly have
\begin{align}
 \hat p_+ &= \frac 1{1+2\frac{\alpha}{m}} \left(	\frac {\sqrt{pq}}{\sqrt{m}}Q^{-1}(P_e/2)+p+\frac{\alpha}{m}\right) \nonumber \\
 \hat p_- &= \frac 1{1+2\frac{\alpha}{m}} \left(	-\frac {\sqrt{pq}}{\sqrt{m}}Q^{-1}(P_e/2)+p+\frac{\alpha}{m}\right) \label{eq:pCLT}
\end{align}
Notice that in general $D(p_1\|p_2)$ is and
increasing function of the distance $|p_1-p_2|$, and therefore
$D(p\|\hat p_\pm)$ for fixed $m$ are increasing functions
of $p$ for $p\in(0,\frac 1 2)$ as the distance $|p-\hat p_\pm|$
is given by $\frac{\sqrt{pq}}{\sqrt{m}}Q^{-1}(P_e/2)\pm\frac\alpha m$. We
can therefore obtain the worst case achievable $a(m,P_e)$ by series
expansion of $D(p\|\hat p_\pm)$ for $p=\frac 1 2$,
\begin{align*}
  D(p\|\hat p_\pm) &= \frac{(p-\hat{p}_\pm)^{2}}{pq\ln4}
  +o\left((p-\hat{p}_\pm)^{2}\right)
\end{align*}
which when inserting (\ref{eq:pCLT}) achieves (\ref{eq:ambounds}).

\emph{Sub-Poisson regime}: In this regime $p=o(\frac 1 m)$. Since $\hat p_-<p$, also
$\hat p_-=o(\frac 1 m)$. The lower tail probability is
\begin{align*}
  P\left(\check p< f^{-1}(\hat p_-)\right) &=
  P\left(\check p< \hat p_-\left(1+2\frac\alpha m\right)-\frac\alpha m\right) \nonumber\\
  & = P\left(k<m\hat p_-\left(1+2\frac\alpha m\right)-\alpha\right)
\end{align*}
For $m$ sufficiently large, the right hand side becomes negative,
and therefore the probability zero. We therefore only have the
constraint
$P\left(\check p> f^{-1}(\hat p_+)\right)\leq P_e$. Write
\begin{align*}
  P\left(\check p> f^{-1}(\hat p_+)\right) &=
  P\left(\check p> \hat p_+ \left(1+2\frac\alpha m\right)-\frac\alpha m\right) \nonumber\\
  & = P\left(k>m\hat p_+\left(1+2\frac\alpha m\right)-\alpha\right)
\end{align*}
If $m\hat p_+\to 0$ the probability converges to one. 
So, let $\hat p_+=\frac{\alpha+\delta}{m}$ with $\delta>0$
arbitrarily small, so that
\begin{align}
  P\left(\check p> f^{-1}(\hat p_+)\right)
  & = P\left(k>(\alpha+\delta)\left(1+2\frac\alpha m\right)-\alpha\right) \label{eq:subP}
\end{align}
As $P(k=0)=(1-p)^m=\left(1-o\left(\frac 1 m\right)\right)^m\to 1$
as $m\to\infty$ the probability (\ref{eq:subP}) 
converges to zero, i.e., is less than $P_e$ for $m$ sufficiently
large so that the constraint is satisfied.
We bound relative entropy by $\chi^2$-distance, see e.g. \cite{DragomirAl01},
\begin{align}
  D(p\|\hat p_+)&\leq \frac{(p-\hat p_+)^2}{\hat p(1-\hat p_+)\ln 2} \nonumber\\
     &= \hat p_+\frac{(1-\frac p{\hat p_+})^2}{(1-\hat p_+)\ln 2} \nonumber\\
     &= \frac{\alpha+\delta}{m\ln 2} +o\left(\frac 1 m\right)
     \label{eq:subPoisson}
\end{align}
because $\frac p{\hat p_+}\to 0$.

\emph{Poisson regime}:
Let $p=\frac \gamma m$. We also set
$\kappa_\pm=m\hat p_\pm$. 
Then
\begin{align}
  D\left(\left.\frac\gamma m\right\|\frac{\kappa_\pm}m\right)
  &= \frac{\kappa_\pm-\gamma+\gamma\ln\gamma-\gamma\ln\kappa_\pm}
     {m\ln 2}+o\left(\frac 1 m\right) \label{eq:DPoisLim}
\end{align}
We define
\begin{align*}
  d(x,y)&=y-x+x\ln\frac x y
\end{align*}
Now
\begin{align*}
  P\left(\check p\leq f^{-1}(\hat p_-)\right) &=
  P\left(\check p\leq \hat p_-(1+2\beta)-\beta\right) \nonumber\\
  & = P\left(k\leq\kappa_-\left(1+\frac {2\alpha}m\right)-\alpha\right)\nonumber\\
  &\to \mathbb{P}_\gamma(\kappa_--\alpha)
\end{align*}
where $\mathbb{P}_\gamma$ is the Poisson CDF. Similarly
\begin{align*}
  P\left(\check p> f^{-1}(\hat p_-)\right) 
  &\to 1-\mathbb{P}_\gamma(\kappa_+-\alpha)
\end{align*}
A reminder about the meaning of $\kappa_\pm$: for every
$\gamma$, $\kappa_-=\sup\left\{\kappa:\mathbb{P}_\gamma(\kappa-\alpha)\leq \frac{P_e}{2}\right\}$, and $\kappa_+=\inf\left\{\kappa:\mathbb{P}_\gamma(\kappa-\alpha)\geq 1-\frac{P_e}{2}\right\}$
Figure \ref{fig:proof} illustrates the proof.
\begin{figure}[hbt]
  \center\includegraphics[width=3.5in]{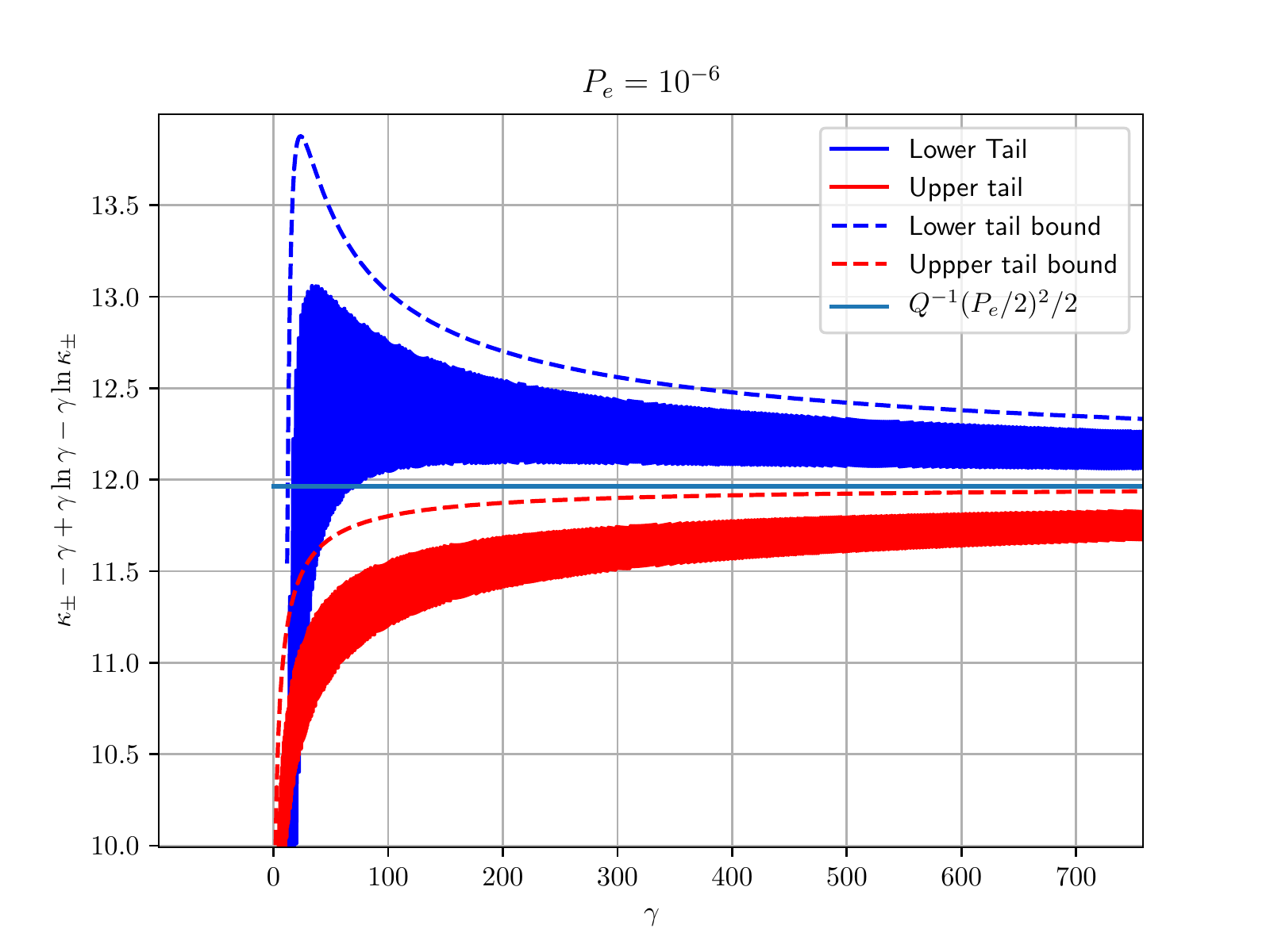}
  \caption{\label{fig:proof}Plot of $d(\tilde\gamma,\tilde\kappa_-),d(\tilde\gamma,\tilde\kappa_+)$ for $P_e=10^{-6}$ for 
  $\alpha=\frac 1 6Q^{-1}(P_e/2)^2-1$. The solid
  curves are for the exact values of $\kappa_\pm$,
  while the dashed curves are the bounds. The solid curves
  are sawtooth like, but this cannot be seen at the scale of
  the figures. The bounds are for the peaks of the solid curves.}
\end{figure}

We will first analyze the lower tail probability
corresponding to $\kappa_-$. 
Let $\gamma_k$, $k=0,1,\ldots$ be the sequence of solutions
 $\mathbb{P}_{\gamma_k}(k)= \frac{P_e}{2}$ -- these
 correspond to the peaks in the solid blue
 curve in Fig. \ref{fig:proof}. Notice that if
$\gamma_{k-1}<\gamma\leq\gamma_k$, $|\gamma-k|\leq|\gamma_k-k|$,
and this is also true for other distance measures. Now,
according to \cite{Short13}\footnote{While \cite{Short13} states the bounds for $k=1,2,\ldots$, it is easy to see that
the bounds are also valid for $k=0$, and the upper bound
is also valid or non-integer values of $k$.},
\begin{align*}
  \frac{P_e}2 = \mathbb{P}_{\gamma_k}(k)
  < \Phi\left(\text{sign}(k+1-\gamma_k)\sqrt{2d(k+1,\gamma_k)}\right)
\end{align*}
so that
\begin{align*}
  \frac 1 2\Phi^{-1}(P_e/2)^2 > d(k+1,\gamma_k)
\end{align*}
(since $k+1<\gamma_k$).
As $\kappa_-=k+\alpha$, we therefore have
\begin{align*}
  d(\kappa_--\alpha+1, \gamma) &< \frac 1 2\Phi^{-1}(P_e/2)^2
   = \frac 1 2Q^{-1}(P_e/2)^2 \nonumber\\
  \gamma-(\kappa_--\alpha+1)+ (\kappa_--\alpha+1)\ln\frac{\kappa_--\alpha+1}\gamma &< \frac 1 2Q^{-1}(P_e/2)^2
\end{align*}
by normalizing by $Q^{-1}(P_e/2)^2$ we get
\begin{align}
  d(\tilde\kappa_--\tilde\alpha_-,\tilde\gamma)=\tilde\gamma-(\tilde\kappa_--\tilde\alpha_-)+ (\tilde\kappa_--\tilde\alpha_-)\ln\frac{\tilde\kappa_--\tilde\alpha_-}{\tilde\gamma} &< \frac 1 2 \label{eq:gameq1}
\end{align}
where specifically $\tilde\alpha_-=\frac{\alpha-1}{Q^{-1}(P_e/2)^2}$. 
From (\ref{eq:DPoisLim}) it can be seen that $a(m,P_e)$ is
determined by the swapped relative entropy,
$d(\tilde\gamma,\tilde\kappa_-)$.
Solving (\ref{eq:gameq1}) with equality we get
\begin{align}
  \tilde\kappa_- -\tilde\alpha_- &= r_-(\tilde\gamma)\tilde\gamma=\frac{\frac 1{2\tilde\gamma}-1}
     {W_{-1}\left(\frac 1 e\left(\frac 1{2\tilde\gamma}-1\right)\right)}
     \tilde\gamma \label{eq:kmsolution}
\end{align}
where $W_{-1}$ is the Lambert $W$-function of order $-1$.

For the upper tail probability, we instead
define $\lambda_k$, $k=0,1,\ldots$ as the sequence of
solutions $\mathbb{P}_{\lambda_k}(k)=1-\frac{P_e}2$.
We use the lower bound from \cite{Short13},
\begin{align*}
  1-\frac{P_e}2 &= \mathbb{P}_{\gamma_k}(k)
  > \Phi\left(\text{sign}(k-\gamma_k)\sqrt{2d(k,\gamma_k)}\right)
  \nonumber\\
    d(k,\gamma_k)&<\frac 1 2\Phi^{-1}(1-P_e/2)^2=\frac 1 2Q^{-1}(P_e/2)^2
\end{align*}
We then have
\begin{align}
  \tilde\gamma-(\tilde\kappa_+-\tilde\alpha_+)+ (\tilde\kappa_+-\tilde\alpha_+)\ln\frac{\tilde\kappa_+-\tilde\alpha_+}{\tilde\gamma} &< \frac 1 2 \label{eq:gameq1p}
\end{align}
where now  $\tilde\alpha_+=\frac{\alpha+1}{Q^{-1}(P_e/2)^2}$. 
The solution of (\ref{eq:gameq1p}) with equality is
\begin{align}
  \tilde\kappa_+ -\tilde\alpha_+ &= r_+(\tilde\gamma)\tilde\gamma =\frac{\frac 1{2\tilde\gamma}-1}
     {W_{0}\left(\frac 1 e\left(\frac 1{2\tilde\gamma}-1\right)\right)}
     \tilde\gamma \label{eq:kpsolution}
\end{align}

The problem is now reduced to finding
\begin{align}
  \sup_{\gamma>0}\max\{d(\tilde\gamma,\tilde\kappa_-),d(\tilde\gamma,\tilde\kappa_+)\}
  &\stackrel{\bigtriangleup}{=} \frac 1 2b(P_e)\label{eq:bPe} \\
  d(\tilde\kappa_--\tilde\alpha_-,\tilde\gamma)&=\frac 1 2 \nonumber\\
  d(\tilde\kappa_+-\tilde\alpha_+,\tilde\gamma)&=\frac 1 2 \nonumber
\end{align}
We notice that $d(\tilde\gamma,\tilde\kappa_-)$ is decreasing
with $\alpha$ while $d(\tilde\gamma,\tilde\kappa_+)$
is increasing. We will first show that if 
$\tilde\alpha_+\leq\tilde\alpha_+^*$ for some $\tilde\alpha_+^*$,
then $\sup_{\gamma} d(\tilde\gamma,\tilde\kappa_+)\leq \frac 1 2$.

By inserting (\ref{eq:kpsolution}) in 
$d(\tilde\gamma,\tilde\kappa_+)$
we find that (Mathematica)
\begin{align}
  \lim_{\tilde\gamma\to\infty}\tilde\kappa_+-\tilde\gamma+\tilde\gamma\ln\frac{\tilde\gamma}{\tilde\kappa_+} = \frac 1 2 \label{eq:dlimt}
\end{align}
We will now prove that if $\tilde\alpha_+\leq\tilde\alpha_+^*$ \emph{and}
 $d(\tilde\gamma,\tilde\kappa_+)>\frac 1 2$, then
 $d(\tilde\gamma,\tilde\kappa_+)$ is increasing in
$\tilde\gamma$, which would then contradict the limit
(\ref{eq:dlimt}).

Let $\tilde\kappa_+=f(\tilde\gamma)$ be the solution (\ref{eq:kpsolution}). The derivative of 
$d(\tilde\gamma,\tilde\kappa_+)$ is
\begin{align*}
f'(\gamma)+\ln\frac{\tilde\gamma}{f(\gamma)}+\tilde\gamma\frac{f'(\gamma)} {f(\tilde\gamma)}\geq 0
\end{align*}
where we want to show the inequality.
The implicit function theorem gives,
\begin{align*}
  f'(\tilde\gamma) &= \frac{\frac{\tilde\kappa+--\tilde\alpha_+}{\tilde\gamma}-1}{\ln \frac{\tilde\kappa+--\tilde\alpha_+}{\tilde\gamma}}
\end{align*}
Thus, to show that $d(\tilde\gamma,\tilde\kappa_+)$  is
increasing we have to show
\begin{align}
  \frac{\tilde\kappa_+ \left(\tilde\gamma \ln \left(\frac{\tilde\gamma}{\tilde\kappa_+}\right) \ln \left(\frac{\tilde\kappa_+-\tilde\alpha_+}{\tilde\gamma}\right)+(\tilde\kappa_+-\tilde\alpha_+)-\tilde\gamma\right)+\tilde\gamma (\tilde\gamma-(\tilde\kappa_+-\tilde\alpha_+))}{\tilde\gamma \tilde\kappa_+ \ln \left(\frac{\tilde\kappa_+-\tilde\alpha_+}{\tilde\gamma}\right)}\geq  0 \label{eq:diff1}
\end{align}
which reduces to
\begin{align*}
  \tilde\kappa_+ \left(\tilde\gamma \ln \left(\frac{\tilde\gamma}{\tilde\kappa_+}\right) \ln \left(\frac{\tilde\kappa_+-\tilde\alpha_+}{\tilde\gamma}\right)\right)\geq
  (\tilde\kappa_+-\tilde\gamma) (\tilde\gamma-(\tilde\kappa_+-\tilde\alpha_+)) 
\end{align*}
As stated above, we assume that
$\tilde\kappa_+-\tilde\gamma+\tilde\gamma\ln\frac{\tilde\gamma}{\tilde\kappa_+}\geq \frac 1 2$, or
$\tilde\gamma\ln\frac{\tilde\gamma}{\tilde\kappa_+}\geq \frac 1 2-(\tilde\kappa_+-\tilde\gamma)$. Inserting, we have
to prove
\begin{align*}
  \tilde\kappa_+ \left(\left(\frac 1 2-(\tilde\kappa_+-\tilde\gamma)\right) \ln \left(\frac{\tilde\kappa_+-\tilde\alpha}{\tilde\gamma}\right)\right)\geq
  (\tilde\kappa_+-\tilde\gamma) (\tilde\gamma-(\tilde\kappa_+-\tilde\alpha)) 
\end{align*}
From (\ref{eq:gameq1p}) with equality we have
\begin{align*}
  (\tilde\kappa_+-\tilde\alpha_+) \ln \left(\frac{\tilde\kappa_+-\tilde\alpha_+}{\tilde\gamma}\right)=
   (\tilde\kappa_+-\tilde\alpha_+)-\tilde\gamma+\frac{1}{2}
\end{align*}
or
\begin{align}
  \tilde\kappa_+ \ln \left(\frac{\tilde\kappa_+-\tilde\alpha}{\tilde\gamma}\right)=
   (\tilde\kappa_+-\tilde\alpha)-\gamma+\frac{1}{2}+\tilde\alpha \ln \left(\frac{\tilde\kappa_+-\tilde\alpha}{\tilde\gamma}\right) \label{eq:kappaeq}
\end{align}
then we have the following sequence of inequalities
\begin{align}
  \left(\left(\frac 1 2-(\tilde\kappa_+-\tilde\gamma)\right) \left((\tilde\kappa_+-\tilde\alpha_+)-\gamma+\frac{1}{2}+\tilde\alpha_+ \ln \left(\frac{\tilde\kappa_+-\tilde\alpha_+}{\tilde\gamma}\right)\right)\right)&\geq
  (\tilde\kappa_+-\tilde\gamma) (\tilde\gamma-(\tilde\kappa_+-\tilde\alpha_+)) \nonumber\\
    \left(\frac 1 2-(\tilde\kappa_+-\tilde\gamma)\right) \left(\frac{1}{2}+\tilde\alpha_+ \ln \left(\frac{\tilde\kappa_+-\tilde\alpha_+}{\tilde\gamma}\right)\right)&\geq
   \frac 1 2(\tilde\gamma-(\tilde\kappa_+-\tilde\alpha_+))\nonumber\\
    \left(\tilde\gamma-\left(\tilde\kappa_+-\frac 1 2\right)\right) \left(\frac{1}{2}+\tilde\alpha_+ \ln \left(\frac{\tilde\kappa_+-\tilde\alpha_+}{\tilde\gamma}\right)\right)&\geq
   \frac 1 2(\tilde\gamma-(\tilde\kappa_+-\tilde\alpha_+))\nonumber\\
    \left(\left(\tilde\kappa_+-\frac 1 2\right)-\tilde\gamma\right) \left(\frac{1}{2}+\tilde\alpha_+ \ln \left(\frac{\tilde\kappa_+-\tilde\alpha_+}{\tilde\gamma}\right)\right)&\leq
   \frac 1 2((\tilde\kappa_+-\tilde\alpha_+)-\tilde\gamma) \nonumber\\
    \tilde\alpha_+\left((r_+(\tilde\gamma)-1)\tilde\gamma+\tilde\alpha_+-\frac 1 2\right) \ln \left(r_+(\tilde\gamma)\right)&\leq
   \frac 1 2\left(\frac 1 {2}-\tilde\alpha_+\right)  \label{eq:derivupper}
\end{align}
($r_+$ is defined in (\ref{eq:kpsolution})).
This is a second order inequality in $\tilde\alpha_+$, which
can be solved to give $\tilde\alpha_+\leq f(\tilde\gamma)$, where
$f(\tilde\gamma)$ is a rather large expression that we will
not write down here. The function
$f(\tilde\gamma)$ is decreasing, and it can be shown (Mathematica)
that $\lim_{\tilde\gamma\to\infty}f(\tilde\gamma)=\frac 1 6$.

We now argue by contradiction. Let $\tilde\alpha_+\leq \frac 1 6$. If at
some time $\tilde\kappa_+-\tilde\gamma+\tilde\gamma\ln\frac{\tilde\gamma}{\tilde\kappa_+}> \frac 1 2$ 
then $\tilde\kappa_+-\tilde\gamma+\tilde\gamma\ln\frac{\tilde\gamma}{\tilde\kappa_+}$ is increasing, thus it stays strictly above $\frac 1 2$. But then the
limit (\ref{eq:dlimt}) cannot be achieved. Thus
we conclude that we must have
$\tilde\kappa_+-\tilde\gamma+\tilde\gamma\ln\frac{\tilde\gamma}{\tilde\kappa_+}\leq \frac 1 2$.

For the lower tail, the above argument can be repeated,
where we now have $\tilde\alpha_-\geq f(\tilde\gamma)$ with
$f(\tilde\gamma)$ the solution of (\ref{eq:derivupper}) with
$r_+(\tilde\gamma)$ replaced with $r_-(\tilde\gamma)$
given by (\ref{eq:kmsolution}). It turns out that in this
case also $\lim_{\tilde\gamma\to\infty}f(\tilde\gamma)=\frac 1 6$.

Since $\tilde\alpha_-<\tilde\alpha_+$ we cannot have both
$\tilde\alpha_-\geq \frac 1 6$ and $\tilde\alpha_+\leq \frac 1 6$.
However, we can choose $\alpha$ so that both 
$\lim_{P_e\to 0}\tilde\alpha_-=\lim_{P_e\to 0}\tilde\alpha_+=\frac 1 6$, and therefore in the limit
both $\sup_{\gamma>0}d(\tilde\gamma,\tilde\kappa_-),\sup_{\gamma>0}d(\tilde\gamma,\tilde\kappa_-)\leq \frac 1 2$.
This shows that $b(P_e)$ given by (\ref{eq:bPe}) converges
to 1 as $m\to\infty$.

To summarize: in the CLT regime we get
the achievable $\lim_{m\to\infty} ma(m,P_e)=\frac{Q^{-1}(P_e)}{2\ln 2}$, in the Poisson regime $\lim_{m\to\infty} ma(m,P_e)=b(P_e)\frac{Q^{-1}(P_e)}{2\ln 2}$ when
$\alpha$ chosen as (\ref{eq:alpha_range})
,
and in the sub-Possion regime (\ref{eq:subPoisson}) $\lim_{m\to\infty} ma(m,P_e)=\frac{\alpha+\delta}{\ln 2}$ with $\delta>0$ arbitrarily small,
which is smaller than $\frac{Q^{-1}(P_e)}{2\ln 2} $
with this choice of $\alpha$. Thus, the Poisson regime
gives the worst performance and results in (\ref{eq:ambounds}).
\end{proof}
The first thing to notice from this result is that as for average
performance, the performance increases as $\frac 1 m$. Specifically,
to beat universal coding of sequences of maximum
length $l$ with  probability $1-P_e$ the number of training
samples is approximately
\begin{align}
  m & \geq  \frac{Q^{-1}(P_e/2)^2}{2\ln 2}\frac l{\log l}
  \label{eq:mboundMarkov}
\end{align}
Comparing this with (\ref{eq:IIDmlimit}) we can see that $m$ still increases
as $\frac l{\log l}$, which is very moderate. However, the
factor in front can be large. Additionally,
the optimum value of $\alpha$ is different; for the
average case $\alpha\approx\frac 1 2$ (\ref{eq:alpha0}), while here
$\alpha\approx \frac 1 6 Q^{-1}(P_e/2)^2$, which can be
considerably larger. The purpose of $\alpha$ is
to avoid that rarely seen symbols require extremely long
codewords. We can therefore interpret this
so that to keep training error very small, protection
against long codewords is of even higher importance then for average performance.

The other thing to remark is that the upper and lower
bounds are only tight in the limit: they are separated
by a factor $b(P_e)$. This factor is determined by
how much the curves in Fig. \ref{fig:proof} "bump"
above the $\frac 1 2 Q^{-1}(P_e/2)$ line. The figure shows
a case where the upper tail probability stays completely
below the line, so that $b(P_e)$ is completely determined
by the bump of the lower tail probability. If one
increases $\alpha$, the lower tail bump becomes smaller and
the upper tail curve will develop a bump. Eventually at
the upper range of (\ref{eq:alpha_range}) the lower tail
probability will be completely below the line, and $b(P_e)$
will be determined completely by the bump on the upper tail
probability. It is clear that if one wants an accurate
value of $b(P_e)$ one can optimize over $\alpha$ in this range,
and one can even use the exact probability rather than the
bounds. But the numerical computation is not easy and might
become unstable for small values of $P_e$. Instead we
suggest to calculate an upper bound on $b(P_e)$ by choosing
$\alpha$ as the lower bound of (\ref{eq:alpha_range}) so that
$b(P_e)$ is determined by the bump of the lower tail,
which can be calculated by
\begin{align}
  b(P_e)&\leq 2\max_{\tilde\gamma>0}\tilde\kappa_--\tilde\gamma+\tilde\gamma\ln\tilde\gamma-\tilde\gamma\ln\tilde \kappa_-
  \label{eq:bPEu}
\end{align}
where $\tilde\kappa_-$ is given by (\ref{eq:kmsolution}),
explicitly
\begin{align*}
  \tilde\kappa_-  &= \frac{\frac 1{2\tilde\gamma}-1}
     {W_{-1}\left(\frac 1 e\left(\frac 1{2\tilde\gamma}-1\right)\right)}
     \tilde\gamma + \frac 1 6 - \frac 2{Q^{-1}(P_e/2)^2} 
\end{align*}
This maximization in (\ref{eq:bPEu}) then easily
can be done numerically. A plot of this upper bound on $b(P_e)$ can be seen in Fig. \ref{fig:Markov} below.

\section{Extension to Markov Chains}
We consider a binary Markov with states $0,1$. Let $p_i=p(i|i)$ be
the probability of staying in state $i$ when the current state is $i$.
The stationary probability is
\begin{align*}
  \pi_i = \frac{p_i}{p_0+p_1}
\end{align*}
Based on training, estimates $\hat p_0$ and $\hat p_1$ are 
generated. The redundancy for coding then is \cite{CoverBook}
\begin{align*}
  \pi_0D(p_0\|\hat p_0)+\pi_1D(p_1\|\hat p_1)
\end{align*}
We consider the two measures of performance
\begin{align}
	R^+_l(m) &=\sup_{p_0,p_1} E\left[\pi_0D(p_0\|\hat p_0)+\pi_1D(p_1\|\hat p_1)\right] \label{eq:RmMarkov} \\
	E(m,a) &=\sup_{p_0,p_1} P\left(\pi_0D(p_0\|\hat p_0)+\pi_1D(p_1\|\hat p_1)\geq a\right) \label{eq:EmaMarkov}
\end{align}
As in the proof of Theorem \ref{thm:iid_binary}, we consider
convergent sequences $(p_0(m),p_1(m))$. From this we
also have a limit $\pi_i=\lim_{m\to\infty}\pi_i(m)$.
If $\pi_0=0$, only the term corresponding to $\pi_1$ in
(\ref{eq:RmMarkov}-\ref{eq:EmaMarkov}) matters; it is therefore reduced
to the iid case, which has better performance. We can therefore
assume that $\pi_i\neq 0$.

In \cite{FalahatgarOrlitsky16,HaoOrlitskyAl18} universal
prediction (called estimation) for Markov chains was considered,
as an extension of \cite{Krichevskiy98}. It was shown
that the estimation error decreases as $\frac{\log\log m}{m}$,
which is an interesting contrast to (\ref{eq:RlIID}).
However, for learned coding the redundancy does not
decrease at all with the length of the training sequence,
\begin{prop}
Assume that the training data consists of a single sequence.
Then
\begin{align*}
   R(m)&\geq \frac 1 2 \nonumber \\
   E(m,a)&=1\quad\text{for }a<\frac 1 2
\end{align*}
\end{prop}
\begin{proof}
Let $p_1=p_0$. The probability that a training sequence
never changes state is then $(1-p_0)^m$. Let us assume
that the training sequence is in state $0$, and that it
uses this to generate a \emph{perfect} estimate $\hat p_0=p_0$.
Since no example of state 1 is seen, the best minimax
estimate of $p_1$ is $\hat p_1=\frac 1 2$. Then
\begin{align*}
  E\left[\pi_0D(p_0\|\hat p_0)+\pi_1D(p_1\|\hat p_1)\right] \geq 
  (1-p_0)^m\frac 1 2 (1-H(p_0))\to \frac 1 2\text{ as }p_0\to 0
\end{align*}
Similarly
\begin{align*}
  P\left(\pi_0D(p_0\|\hat p_0)+\pi_1D(p_1\|\hat p_1)\geq \frac 1 2 (1-H(p_0))\right)\geq (1-p_0)^m
\end{align*}
\end{proof}
The different behavior is easy to understand intuitively.
If the Markov chain is slow mixing, the training sequence
might see only one state. But the test sequence could
be from the other state, about which nothing has been
learned. On the other hand, in universal prediction, the
goal is to predict the next sample after a training sequence.
If the training sequence has seen only a single state,
there is a good chance the following sample is also from that state, so prediction is not too bad.
The point is that universal prediction and learned coding
are not quite equivalent.

From the above it is clear that multiple training sequences are required for learning how to code.
For achievability, 
we let the training data consist of $n$ sequences of length $l$.
We assume that each sequence 
has an initial state according to the stationary
distribution. Let $\hat m_i$ be the number of visits to state $i$;
this is of course random, but converges towards $\pi_i m$ by
the law of large numbers. In order to handle this randomness
in the estimators, we consider genie assisted estimators. 
For achievability we consider an estimator that is \emph{inhibited} by
the genie as follows. The genie fixes $m_i$ with $m_1+m_0\leq m-n$ and inhibits the estimation as follows
\begin{itemize}
  \item If there are more than $m_i$ visits to state $i$,
        it only uses the first $m_i$ ones to estimate $p_i$.
  \item If there is either less than $m_0$ visits to state $0$ or
        less than $m_1$ visits to state $1$, the genie generates
        more training data with the correct distribution until there is the correct number of
        visits; however, it also marks the sequence as invalid.
\end{itemize}
In either case we use the estimator (\ref{eq:addestimator})
to estimate
\begin{align}
  \hat p_i &= \frac{k_{ii}+\alpha}{m_i+2\alpha} \label{eq:addMarkov}
\end{align}
where $k_{ii}$ is the number of times the sequence
stays in state $i$ out of the $m_i$ visits to state $i$. 
This estimator is of course not realizable; a realizable
estimator is one that uses (\ref{eq:addMarkov}) based
on the actual number of visits $\hat m_i$. But the
way we have constructed the genie-inhibited estimator means
that the realizable estimator has at least as good performance.

For the converse, we assume that the genie makes
training data that has exactly
$m_i$ visits to state $i$ with $m_1+m_0=m$. We optimize the estimator
over $m_1,m_0$; thus there is no assumption that sequences
start according to the stationary distribution.
 
Denote by $E_2$ the event that the genie marks
a sequence as invalid. We now bound $P(E_2)$. Let
\begin{equation}
  m_i = (\pi_i-\epsilon) (m-n) \label{eq:mi}
\end{equation}
 and let
$\hat m_i$ be the actual number of visits to state $i$.
Then
\begin{align}
  P(E_2) & = P(\hat m_0 < m_0 \vee \hat m_1 < m_1) \label{eq:PE2}
\end{align}
We then have 
\begin{lem}\label{thm:PE2bound}
With $m_i$ given by (\ref{eq:mi}) and $P(E_2)$ by (\ref{eq:PE2}) we can bound
\begin{align*}
  P(E_2)&\leq 2\exp(-2n\epsilon^2)
\end{align*}
for all $p_0, p_1$.
\end{lem}
\begin{proof}
Let $S_i$ be the number of visits to state 0 in the $i$-th
training sequence; the total number of visits then is
$S=\sum_{i=1}^nS_i$, with the $S_i$ independent. Independent of $p_0$ and $p_1$ we have
$0\leq S_i \leq l-1$. We can write
\begin{align}
  P(E_2) & = P(S < (\pi_0 -\epsilon)(m-n)
    \vee m-l-S < (1-\pi_0 -\epsilon)(m-n)) \nonumber\\
  &=  P(S < (\pi_0 -\epsilon)(m-n))+P((\pi_0+\epsilon)(m-n)<S)
  \label{eq:PE2bound}
\end{align}
Since we know the range of the $S_i$ and they are independent,
we can conveniently use Hoefding's inequality \cite{GrimmettBook},
\begin{align*}
  P(S < (\pi_0 -\epsilon)(m-n))
  &= P\left(\frac 1 n(S-\pi_0(m-n)) < -\epsilon (l-1)\right) \nonumber\\
  &\leq \exp\left(-2\frac{n\epsilon^2(l-1)^2}{(l-1)^2}\right)
\end{align*}
The other probability in (\ref{eq:PE2bound}) is bound similarly.
\end{proof}

\begin{thm}\label{thm:MarkovAverage}
Consider a binary Markov chain. Assume that the training consists of a set
of sequences, so that both the size of the set and the length of the sequences 
approach infinity. For the estimator (\ref{eq:addestimator}),
with $\alpha=\alpha_0$ (\ref{eq:alpha0}) we get
\begin{align}
  R_l^+(m) &=  \frac{2\alpha_0}{m\ln 2}+o\left(\frac 1 m\right) 
\end{align}
while a lower bound is
\begin{align}
  R_l^+(m) &\geq  \frac{1}{m\ln 2}+o\left(\frac 1 m\right) 
\end{align}
\end{thm}
\begin{proof}
We use the genie-inhibited estimator, so that
the estimator always uses exactly $m_i$ samples to
estimate $p_i$. Whenever $E_2$ happens, we add a penalty of
$1$ to the codelength and thereby get an upper bound. So,
\begin{align}
	\lim_{m\to\infty} m R_l^+(m) &=\lim_{m\to\infty}m\sup_{p_0,p_1} E\left[\pi_0D(p_0\|\hat p_0)+\pi_1D(p_1\|\hat p_1)\right] \nonumber \\
	&\leq \pi_0\lim_{m\to\infty}\sup_{p_0} \frac{m}{m_0}m_0 E\left[D(p_0\|\hat p_0)\right] \nonumber \\
	&+\pi_1\lim_{m\to\infty}\sup_{p_1} \frac{m}{m_1}m_1 E\left[D(p_1\|\hat p_1)\right] \nonumber\\
	&+\lim_{m\to\infty}mP(E_2) \nonumber\\
	&\leq \lim_{m\to\infty}\frac{\pi_0}{(\pi_0-\epsilon)\left(1-\frac 1 l\right)}\sup_{p_0}  E\left[D(p_0\|\hat p_0)\right] \nonumber \\
	&+\pi_1\lim_{m\to\infty}\frac{\pi_1}{(\pi_1-\epsilon)\left(1-\frac 1 l\right)}\sup_{p_1}  E\left[D(p_1\|\hat p_1)\right] \nonumber\\
	&+\lim_{m\to\infty}2m\exp(-2n\epsilon^2) \label{eq:limmR} \\
	&=2\alpha_0
\end{align}
The condition for the first two terms in (\ref{eq:limmR}) to
converge to $\alpha_0$ is just that $l\to\infty$ and
$\epsilon\to 0$; there is no requirement on the
rate of convergence. The condition for the last
term to converge to zero is just that $n\to\infty$ and
that $\epsilon$ does not converge to zero too fast. We
can always choose a suitable $\epsilon$ to satisfy this.

For the converse, we use the idealized estimator
\begin{align}
	\lim_{m\to\infty} m R(m) &=\lim_{m\to\infty}\inf_{m_0}m\sup_{p_0,p_1} E\left[\pi_0D(p_0\|\hat p_0)+\pi_1D(p_1\|\hat p_1)\right] \nonumber \\
	&\geq \inf_{m_0}\sup_{\pi_0}\left\{\lim_{m\to\infty}\sup_{p_0}\pi_0\frac{m}{m_0} m_0 E\left[D(p_0\|\hat p_0)\right]\right. \nonumber \\
	&+\left.\lim_{m\to\infty}\sup_{p_1}\pi_1\frac{m}{m_1} m_1 E\left[D(p_1\|\hat p_1)\right]\right\} \nonumber\\
	&\geq \frac 1 2\inf_{\hat \pi_0}\sup_{\pi_0}\frac{\pi_0}{\hat\pi_0}
	+\frac{1-\pi_0}{1-\hat\pi_0} \label{eq:pi0}
\end{align}
where $\hat p_0$ stands for $\frac{m_0}{m}$. The last
inequality is due to \cite[Theorem 2]{Krichevskiy98}.
The minimax
problem in (\ref{eq:pi0}) is easily solved through differentiation,
with the result that the minimax solution is $\hat\pi_0=\pi_0=\frac 1 2$, which gives $\lim_{m\to\infty} m R(m)\geq 1$.
\end{proof}
The result could potentially be improved by considering
the modified estimator from \cite{BraessSauer04} that resulted in
(\ref{eq:RlIIDexact}).

\begin{thm}\label{thm:MarkovPe}
Consider a binary Markov chain. Assume that the training consists of a set
of sequences, so that both the size of the set and the length of the sequences 
approach infinity. 
	Using the estimator from Theorem \ref{thm:iid_binary}, the
following decay is achievable
\begin{align}
  a(m,P_e) &= 2b(P_e)\frac{Q^{-1}((1-\sqrt{1-P_e})/2)^2}{2m\ln 2}+o\left(\frac 1 m\right) \label{eq:amMarkov}
\end{align}
This is achievable for
\begin{align*}
 \frac 1 6 Q^{-1}((1-\sqrt{1-P_e})/2)^2 -1\leq \alpha \leq \frac 1 6 Q^{-1}((1-\sqrt{1-P_e})/2)^2 +1
\end{align*}
\end{thm}
\begin{proof}
%
%
Define the event
\begin{align*}
E_1&=\{\pi_0D(p_0\|\hat p_0)+\pi_1D(p_1\|\hat p_1)\geq a\} 
\end{align*}
While $E_2$ is still the event that the genie
flags a sequence as invalid.
These events are not independent, but we can upper bound
the probability of training error by
\begin{align*}
  E(m,a)&\leq \sup_{p_0,p_1}(P(E_1)+P(E_2))
\end{align*}
Of course, the addition of extra artificial data can
decrease $P(E_1)$, but whenever artificial data is added $E_2$ happens,
and the total error probability is not decreased.
Here
\begin{align*}
P(E_1) &=	1-P(\pi_0D(p_0\|\hat p_0)+\pi_1D(p_1\|\hat p_1)\leq a) \nonumber\\
&\leq 1-P\left(D(p_0\|\hat p_0)\leq \frac a{2\pi_0} \wedge D(p_1\|\hat p_1)\leq \frac a{2\pi_1} \right) \nonumber \\
&= 1-P\left(D(p_0\|\hat p_0)\leq \frac a{2\pi_0} \right)P \left(D(p_1\|\hat p_1)\leq \frac a{2\pi_1}\right)
\end{align*}
because the random variables $\hat p_0$ and $\hat p_1$ are
independent with the way the (augmented) data set is generated.
We now require
\begin{align}
  P\left(D(p_i\|\hat p_i)\geq \frac a{2\pi_i} \right)
  \leq 1-\sqrt{1-P_e+P(E_2)} \label{eq:MCPD}
\end{align}
Notice that by construction, $\hat p_i$ has exactly
the same distribution as $\hat p$ in Theorem \ref{thm:iid_binary},
except based on $m_i$ samples instead of $m$.
By Theorem \ref{thm:iid_binary} (\ref{eq:MCPD}) is satisfied if
\begin{align*}
  \frac a{2\pi_i}\geq \frac{b(P_e)Q^{-1}((1-\sqrt{1-P_e+P(E_2)})/2)^2}{2m_i\ln 2}+o\left(\frac 1 m\right).
\end{align*}
With $m_i$ given by (\ref{eq:mi}) and $P(E_2)$
bounded by Lemma \ref{thm:PE2bound}, as long as
both $n,l\to\infty$, we can let $\epsilon$ converge to
0 sufficiently slowly so that we get (\ref{eq:amMarkov}).
\end{proof}

\begin{thm}
	  A lower bound is
	\begin{align}
  a(m,P_e) &\geq \frac{F_{\chi^2_2}^{-1}(1-P_e)}{2m\ln 2}+o\left(\frac 1 m\right) 
\end{align}
where $F_{\chi^2_2}$ is the CDF for a $\chi^2$-distribution
with two degrees of freedom.
\end{thm}
\begin{proof}
The proof follows closely the proof of the lower bound
in Theorem \ref{thm:iid_binary}, and we only
provide the outline here. We do the lower bound
for the CLT regime, and we
use (\ref{eq:Pinsker}) to lower bound relative entropy,
\begin{align*}
  \pi_0D(p_0\|\hat p_0)+\pi_1D(p_1\|\hat p_1)
  \geq \frac 2{\ln 2}\left(\pi_0(\hat p_0-p_0)^2+\pi_1(\hat p_1-p_1)^2\right)
\end{align*}
We have $\hat p_i=f_i(\check p_0)$, which we assume
is invertible; we write $\check p_i=f^{-1}(\hat p_i)
=\hat p_i+g_{i,m}(\hat p_i)$. Let
\begin{align*}
  S_a =\left\{(\hat p_0,\hat p_1): \frac 2{\ln 2}\left(\pi_0(\hat p_0-p_0)^2+\pi_1(\hat p_1-p_1)^2\right)\leq a\right\}
\end{align*}
which is an ellipsoid centered at $(p_0,p_1)$. Consider the
set of points satisfying
\begin{align*}
  \left(\sqrt{m}(\check p_0-p_0),\sqrt{m}(\check p_1-p_1)\right)
  \in \sqrt{m}(S_a-(p_0,p_1))+\sqrt{m}g_m(S_a)
\end{align*}
 Here $\sqrt{m}(\check p_i-p_i)$ converge to independent Gaussians
 by the central limit theorem. To achieve a given
 $P_e$ we must then have a decrease as $\frac 1 m$, and
 $S_a$ is a shrinking ellipsoid. As in the proof of
 Theorem \ref{thm:iid_binary} we see that $\sqrt{m}g_m(S_a)$
 converges to a single point, $b$. The goal is to minimize
 $a$ while still having $P(S_a)\geq 1-P_e$. It is easily seen
 that this is achieved for $b=0$. Finally the worst
 case over $p_0,p_1$ is for $p_0=p_1=\frac 1 2$ 
 while the best case is $m_i=\pi_i m$. We end up
 with analyzing the symmetric case,
$\check p_i-p_i\sim \mathcal{N}\left(0,\frac 1 {2m}\right)$,
so that
\begin{align*}
 2m\ln 2 \frac 2{\ln 2}\left(\frac 1 2(\hat p_0-p_0)^2+\frac 1 2(\hat p_1-p_1)^2\right) \sim \chi^2_2
\end{align*}
and
\begin{align*}
  a(m,P_e)\geq \frac{F_{\chi^2_2}^{-1}(1-P_e)}{2m\ln 2}
\end{align*}
\end{proof}
As opposed to the IID case, Theorem \ref{thm:iid_binary}, the
upper and lower bounds are not tight as $P_e\to 0$. There
is a factor about 2 between the bounds. Both bounds
require evaluations of the probability of
a set $\{\hat p_0, \hat p_1:\pi_0D(p_0\|\hat p_0)+\pi_1D(p_1\|\hat p_1)\leq a\}$.
The achievable bound bounds this by the probability of
a square, while the converse uses a circle, and the gap
is due to this difference. Fig. \ref{fig:Markov} shows
the different bounds. The achievable bounds indicates
that the Markov chain requires more than twice as
many training samples as the iid case; but the lower bound
shows that hardly any increase in samples is needed (though
at least it proves that some more samples are required).

However, our bottom-line comparison was with universal
source coding. The redundancy of universal source
coding of a Markov chain with 2 states is about 
$R_l^+\approx \frac{\log l}{l}$ \cite{Rissanen86b}, a factor
2 increase over IID sources. For the achievable rate we also
have about a factor 2 increase, and therefore approximately
\begin{align*}
  m & \geq  \frac{Q^{-1}(P_e/2)^2}{2\ln 2}\frac l{\log l}
\end{align*}
the same as (\ref{eq:mboundMarkov}). Thus no more
samples are required than for the IID case. On the other hand,
if we go with the lower bound, the factor for training is closer to 1, and
the conclusion is therefore that only half the number of samples
is required to beat universal coding for Markov chains compared
to the IID case: it is even easier to learn Markov chains than
IID sources! The author does not really have a conjecture
on whether the upper or lower bound is tighter. 

\begin{figure}[hbt]
\center\includegraphics[width=3.5in]{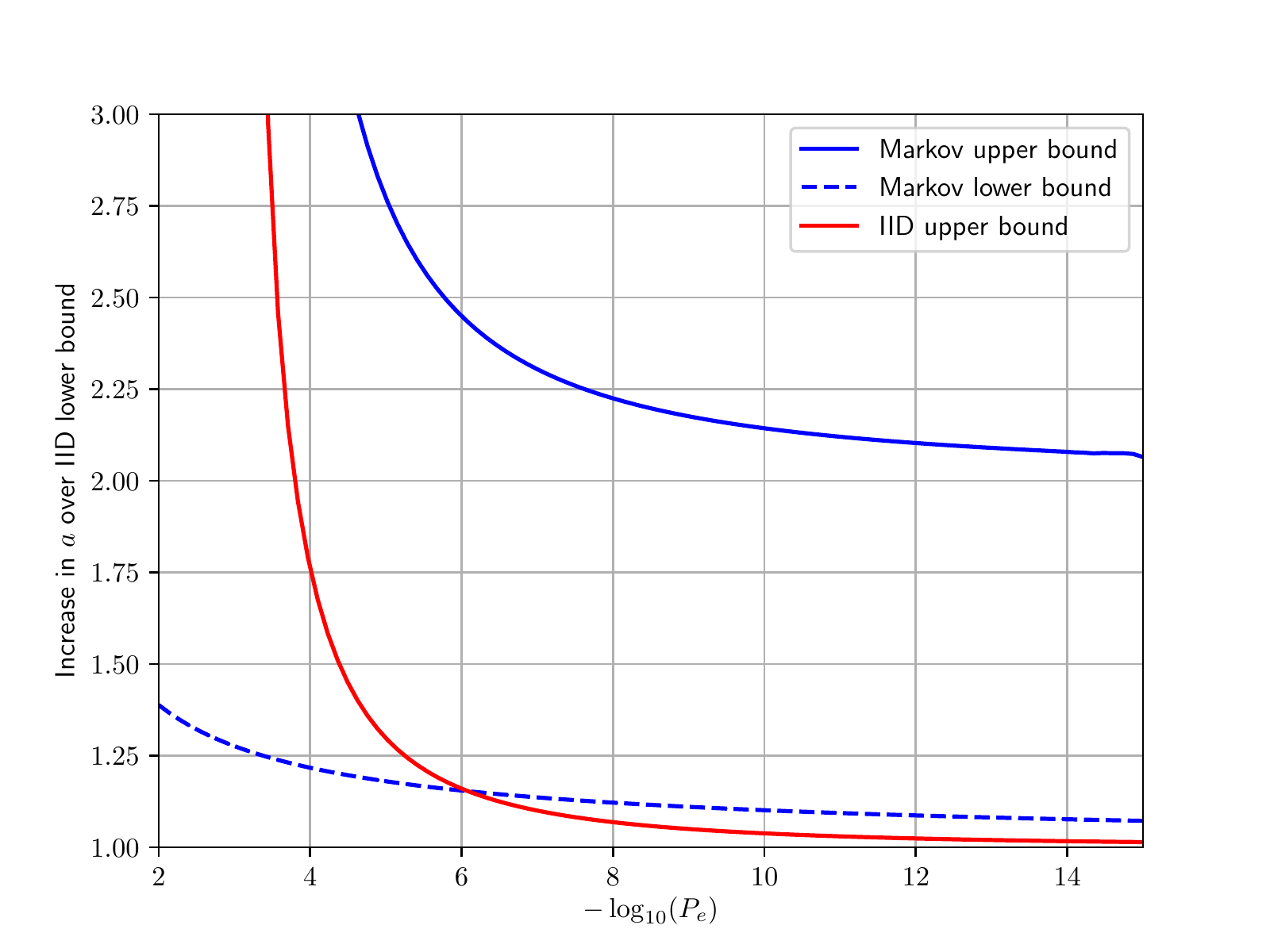}
  \caption{\label{fig:Markov}Plot of the gap
  between the IID lower bound and the Markov bounds.}
\end{figure}

Finally, the condition for achievability in both Theorems
\ref{thm:MarkovAverage} and \ref{thm:MarkovPe} is just
that both the number of training sequences $n$ and their length $l$
approach infinity,
but it does not matter how. One would think that there would
be an optimum relationship between $n$ and $l$, but obviously
that does not affect the performance to the first order.

\section{Conclusions }
The central question of this paper can be thought of as: how
much training is required to beat universal source coding? The
answer for both IID sources and Markov chains is: not many.
To code a sequence of length $l$ the number of training
samples is proportional to $\frac{l}{\log l}$. 
This optimistic conclusion is totally opposite to the pessimistic
conclusion of \cite{HershkovitsZiv97}. The reason is due to
the viewpoint -- and perhaps that
we so far only consider very simple sources. While \cite{HershkovitsZiv97} focuses
on approaching entropy rate, we just want to beat the
redundancy of universal source coding. Additionally, \cite{HershkovitsZiv97}
considers learning to be that of building a dictionary,
inspired by Lempel-Ziv coding \cite{ZivLempel77,ZivLempel78}.
However, the exceptional performance of modern machine
learning can be seen as being achieved through learning
soft information. This is also how we approach
learning here, and the coding is more similar to the
CTW algorithm of Willems \cite{WillemsAl95}.

The paper
also shows that it is essential to have multiple
training sequences for Markov chains, which is different
from universal prediction.

Natural generalization of the result is to larger alphabet
sizes and Markov chains with more than two states. The most
interesting generalization is to finite state machines (FSM) as
in
\cite{Rissanen86b}. FSM are much more realistic of real-world
sources, and potentially they might not be so easy to learn.

Another interesting generalization would
be to the large deviations regime. In this case $a$
would be fixed and the question then is how $P_e$ decreases
with $m$. This seem to be more difficult than the moderate
deviations considered here. One thing we can already predict from
the results here is that in the large deviation regime the
"add $\alpha$" estimator (\ref{eq:addestimator}) is no longer
suitable. Namely, from (\ref{eq:alpha_range}) it can
be seen that the optimum $\alpha$ depends on $P_e$,
which again depends on $m$ for large deviations. Thus, $\alpha$
must also be a function of $m$. In fact, from some initial
numerical observations, it seems the estimator (\ref{eq:addbeta})
with $\beta$ constant works much better.

\bibliographystyle{IEEEtran}
\bibliography{Coop06,ahmref2,Coop03,BigData,combined,ECGandHRV}

%
%
\end{document}